\documentclass{amsart}[11pt]
\usepackage{amsmath, amsfonts, amsthm, amssymb,mathtools}
\usepackage{subfigure}
\usepackage{caption} 
\usepackage{stmaryrd}
\usepackage{verbatim}
\usepackage{hyperref}
\usepackage{color}
\usepackage{ulem}
\usepackage{enumerate}
\numberwithin{equation}{section}
\newtheorem{theorem}{Theorem}[section]

\newtheorem{lemma}[theorem]{Lemma}
\newtheorem{proposition}[theorem]{Proposition}

\theoremstyle{definition}
\newtheorem{definition}[theorem]{Definition}
\newtheorem{remark}[theorem]{Remark}

\newcommand{\ind}{1\hspace{-2.1mm}{1}} 

\newcommand{\RR}{\mathbb{R}}
\newcommand{\wt}{\widetilde}
\newcommand{\PP}{\mathbb{P}}

\newcommand{\NN}{\mathbb{N}}
\newcommand{\PPt}{\widetilde{\mathbb{P}}}
\newcommand{\D}{\mathrm{d}}
\newcommand{\Dd}{\mathcal{D}}
\newcommand{\E}{\mathrm{e}}
\newcommand{\BS}{\mathrm{BS}}

\newcommand{\EE}{\mathbb{E}}

\newcommand{\mm}{\mathrm{m}}
\newcommand{\sgn}{\mathrm{sgn}}

\newcommand{\Nn}{\mathcal{N}}
\newcommand{\Ff}{\mathcal{F}}
\newcommand{\Oo}{\mathcal{O}}
\newcommand{\urm}{\mathrm{u}}

\begin{document}
\title{Asymptotic behaviour of the fractional Heston model}
\date{\today}
\author{Hamza Guennoun}
\address{Soci\'et\'e G\'en\'erale, Global markets Quantitative Research}
\email{guennoun.hamza@sgcib.com}
\author{Antoine Jacquier}
\address{Department of Mathematics, Imperial College London}
\email{a.jacquier@imperial.ac.uk}
\author{Patrick Roome}
\address{JP Morgan}
\email{patrick.roome@gmail.com}
\author{Fangwei Shi}
\address{Department of Mathematics, Imperial College London}
\email{fangwei.shi12@imperial.ac.uk}
\thanks{The authors would like to thank Jim Gatheral for useful discussions and Mathieu Rosenbaum for his encouragement.
AJ acknowledges financial support from the EPSRC First Grant EP/M008436/1.
FS is funded by a mini-DTC scholarship from the Department of Mathematics, Imperial College London.
}
\subjclass[2010]{60F10, 91G99, 91B25}
\keywords{Stochastic volatility, implied volatility, asymptotics, fractional Brownian motion, Heston model}

\begin{abstract}
We consider the fractional Heston model originally proposed
by Comte, Coutin and Renault~\cite{CCR12}.
Inspired by recent ground-breaking work on rough volatility~\cite{ALV07, BayerRough, Fuka, GJR14}
which showed that models with volatility driven by fractional Brownian motion with short memory 
allows for better calibration of the volatility surface
and more robust estimation of time series of historical volatility,
we provide a characterisation of the short- and long-maturity asymptotics of the implied volatility smile.
Our analysis reveals that the short-memory property precisely provides a jump-type behaviour of the smile
for short maturities, thereby fixing the well-known standard inability of classical stochastic volatility models to fit 
the short-end of the volatility smile.
\end{abstract}

\maketitle


\section{Introduction}
Since the Black-Scholes model was introduced forty years ago, practitioners and academics
have been proposing refinements thereof in order to take into account the specific behaviour of market data.
In particular, stochastic volatility models, turning the constant Black-Scholes instantaneous volatility of returns
into a stochastic process, have been studied and used heavily.
Monographs such as \cite{FouqueBook, GatheralBook, GuliBook, PHLBook, LewisBook, RebonatoBook} are great sources of information regarding this large class of models, 
both from a theoretical point of view (existence and uniqueness of these processes, asymptotic behaviour), 
and with practitioner's insights (how these models actually perform, how they should behave compared to market data).
Despite the success of these models, it is now widely understood that calibration of the observed implied 
volatility surface fails for short maturities, the observed smile being steeper than that generated by diffusions
with continuous paths.
To remedy this issue, several authors have suggested the addition of jumps, 
either in the form of an independent L\'evy process~\cite{Bates} or within the more general framework of
affine processes~\cite{JKRM, KR}.
Jumps (in the stock price dynamics) imply an explosive behaviour for the implied volatility for short maturities
(see~\cite{Tankov} for a review of this phenomenon), but are able to capture the observed steepness of 
the implied volatility.
This could be the end of the modelling story; 
however, this approach has also had (and still has) his share of controversy since the jump part of the process
is notoriously difficult to hedge, making its practical implementation a rather delicate (and sometimes philosophical) issue.

From a time series modelling point of view, classical stochastic volatility models, 
driven by a Brownian motion, have been criticised for not taking into account the long memory 
of the observed volatility of returns.
In ARCH and GARCH models, memory (quantified through the autocorrelation function) decays exponentially fast, 
whereas it does not decay at all for integrated versions of these.
In the discrete-time setting, this has led to the introduction of fractionally integrated models,
such as ARFIMA~\cite{Granger} and FIGARCH~\cite{Baillie}.
In continuous time, this long-memory behaviour has been modelled through 
fractional Brownian motion~\cite{CR98, CCR12} with Hurst exponent strictly greater than $1/2$.
Fractional Brownian motion has its pitfalls though, since it is not a semimartingale, 
and yields arbitrage opportunities~\cite{Rogers, Shiryaev}.
This can be avoided using heavy machinery~\cite{Cheridito, Elliott, Guasoni, Hu}, 
but as a by-product introduces non-desirable economic features~\cite{Bjork}.
As documented in~\cite{Vilela}, these issues are however not relevant when the fractional Brownian
motion drives the instantaneous volatility rather than the stock price itself.

These fractional stochastic volatility models, somehow popular in the econometrics and statistics communities, have however received little attention from more classical mathematical finance and stochastic analysis groups.
Gatheral, Jaisson and Rosenbaum~\cite{GJR14} have recently suggested to consider the Hurst exponent, 
not as an indicator of the historical memory of the volatility, but as an additional parameter to be calibrated
on the volatility surface.
Their study reveals that $H\in (0,1/2)$, which seems to indicate short memory of the volatility,
thereby contradicting decades of time series analyses.
By considering a specific fractional volatility model, 
directly inspired by the fractional version of the Heston model~\cite{CCR12, Heston}, 
we provide a theoretical justification of this result.
We show in particular that, when $H\in (0,1/2)$, the implied volatility in this model 
explodes in the jump-sense.
Probabilistically speaking, this means that a rescaled (by time) version of the log stock price
process converges weakly, but not the process itself, 
which is reminiscent of what happens in the case of jump-diffusions.
In the case $H\in (1/2, 1)$, heuristically, long memory does not have time to affect the dynamics
of the process and the implied volatility converges to a strictly positive constant.
For large maturities, the phenomenon is reversed, in the sense that short memory gets somehow 
averaged out and the behaviour of the implied volatility is similar to the standard Brownian-driven diffusion case, 
whereas long memory yields a different asymptotic behaviour.
Finally, we comment that recently another fractional version of the Heston model was proposed and analysed in~\cite{Euch1,Euch2,Euch3}, 
in which the authors defined a different structure for the variance process through fractional integration.
In particular, El Euch, Fukasawa and Rosenbaum~\cite{Euch3} bridge the connection between market microstructure and rough volatility by proposing a microscopic price model and by showing that it converges to a rough Heston setting in the long term.
We also refer interested readers to~\cite{Alos17, BayerRough, BayerRoughSkew, forde16} for more recent developments on 
fractional volatility modelling.

This paper is structured as follows:
in Section~\ref{sec:Model}, we introduce the model and study its main properties. 
We in particular show that the characteristic function of the stock price 
is available in closed form.
In Section~\ref{sec:Asymptotics}, we derive the main, probabilistic and financial, results of the paper, 
namely large deviations principles for rescaled versions of the log stock price process and the
asymptotic behaviour of the implied volatility, both for short and for large maturities.
In Section~\ref{sec:numerics} we provide several numerical examples.
Section~\ref{sec:Proofs} contains the proofs of the main results, 
and we add a short reminder on large deviations in the appendix.

\section{The Model and its properties}\label{sec:Model}
\subsection{The fractional Heston model}
Let $(\Omega,\Ff, (\Ff_t)_{t\geq 0}, \PP)$ be a given filtered probability space supporting two independent standard Brownian motions $B$ and $W$.
We denote by $(S_t)_{t\geq 0}$ the stock price process, and let $X_t:=\log(S_t)$ 
be the unique strong solution to the stochastic differential equation
\begin{equation}\label{eq:fhm}
\begin{array}{rll}
\D X_t &= \displaystyle -\frac{1}{2}V^d_t \D t +\sqrt{V^d_t}\D B_t,
\qquad & X_0 = 0,\\
\D V_t &= \displaystyle \kappa (\theta -V_t)\D t+\xi \sqrt{V_t}\D W_t,
\qquad & V_0 = v_0>0,\\
V^d_t &= \displaystyle \eta + I^d_{0+}V_t,
\end{array}
\end{equation}
where $d\in (-1/2, 1/2)$ and the coefficients satisfy $\kappa, \theta, \xi>0$.
The operator~$I^d_{0+}$ is 
the classical left fractional Riemann-Liouville integral of order~$d$:
\begin{equation*}
I^d_{0+}\phi(t) 
:= 
\left\{
\begin{array}{ll}
\displaystyle \int_{0}^{t}\frac{(t-s)^{d-1}}{\Gamma(d)}\phi(s) \D s,
 & \quad \displaystyle \text{for }d\in \left(0,\frac{1}{2}\right),\\
\displaystyle \frac{\D}{\D t}I_{0+}^{d+1}\phi(t),
 & \quad \displaystyle \text{for }d\in \left(-\frac{1}{2}, 0\right),
\end{array}
\right.
\end{equation*}
valid for any function $\phi\in L^1([0,t])$,
where $\Gamma$ is the standard Gamma function.
For more details on these integrals, we refer the interested reader to the monograph~\cite[Chapter~1, Section~2]{Samko},
and for the application to discretisation schemes, we refer to~\cite[Section~5]{CCR12}.
The couple $(X,V)$ corresponds to the standard Heston stochastic volatility model~\cite{Heston},
which admits a unique strong solution by the Yamada-Watanabe conditions~\cite[Proposition 2.13, page 291]{KS97}).
The additional parameter $\eta\geq 0$, as explained in~\cite{CCR12} allows to loosen the tight connection between the mean and the variance of $V_t$.
The process $V$ can be written, in integral form, as 
$V_t = v_0\E^{-\kappa t} + \theta\left(1-\E^{-\kappa t}\right) + \xi\int_{0}^{t}\E^{-\kappa(t-s)}\sqrt{V_s}\D W_s$,
and therefore
It\^o isometry implies that the covariance structure reads, for any $t,u>0$, 
$$
\langle V_u, V_t\rangle  = \frac{\xi^2\theta}{2\kappa}\E^{-\kappa|t-u|} + \frac{\xi^2}{\kappa}(v_0-\theta)\E^{-\kappa(t\wedge u)} - \frac{\xi^2}{2\kappa}(2v_0-\theta)\E^{-\kappa(t+u)}.
$$
The Feller condition~\cite[Chapter 15]{KT81}, $2\kappa\theta\geq\xi^2$, ensures that the origin is unattainable (otherwise it is regular, hence attainable, and strongly reflecting); 
under this condition, since the Riemann-Liouville operator preserves positivity, 
$V_t^d\geq \eta$ almost surely for all $t\geq 0$.
Now, for any $t\geq 0$, we have
$\EE(V_t^d) = \eta + I^d_{0+}\left(v_0\E^{-\kappa t} + \theta\left(1-\E^{-\kappa t}\right)\right)$,
and, for any $t, h\geq 0,$
$$
\langle V_{t+h}^d, V_{t}^d\rangle = 
\int_{0}^{t+h}\int_{0}^{t}\frac{(t+h-s)^{d-1}(t-u)^{d-1}}{\Gamma(d)^2}\langle V_s, V_u\rangle\D u \D s.
$$
The motivation for such a fractional volatility model comes from the seminal work by 
Comte and Renault~\cite{CR98} on the Ornstein-Uhlenbeck process.
Consider the unique strong solution to the stochastic differential equation
$\D X(t)= -\kappa X(t) \D t +\sigma \D W^d (t)$ starting at zero,
with $d \in (0,1/2)$, $\kappa,\sigma>0$ and $W^d$ a fractional Brownian motion;
its fractional derivatives $X^{-d}$ defined via the identity 
$X_t = I^{1+d}_{0+}X^{-d}(t)$, for all $t\geq 0$, or
\begin{equation*}
X^{-d}(t)=\frac{\D}{\D t}\left[ \int^t _0 \frac{(t-s)^{-d}}{\Gamma(1-d)} X_s \D s   \right]= \int^t _0 \frac{(t-s)^{-d}}{\Gamma(1-d)}  \D X_s
\end{equation*}
satisfies the SDE
$\D X^{-d}(t)= -\kappa X^{-d}(t) \D t +\sigma \D B (t)$,
with $X^{-d}(0)=0$, where $B$ is a standard Brownian motion.
Conversely, if $X(t)$ satisfies the Ornstein-Uhlenbeck SDE
$\D X(t)= -\kappa X(t) \D t +\sigma \D W (t)$,
with $X(0)=0$,
then 
\begin{equation*}
\D I^{1+d}_{0+}X(t)= -\kappa I^{1+d}_{0+}X(t) \D t +\sigma \D W^d (t),
\quad\text{with }I^{1+d}_{0+}X(0) = 0.
\end{equation*}
Finally, note that the fair strike of a continuously monitored variance swap with maturity~$T$ reads
\begin{align}\label{eq:VS}
\frac{1}{T}\EE\left(\int_0^TV_t^d\D t\right)
&= \frac{1}{T}\int_0^T\left\{\eta + I_{0+}^d\left[v_0\E^{-\kappa t} + \theta(1-\E^{-\kappa t})\right]\right\}\D t\nonumber\\
&= \eta + \frac{v_0T^d}{\Gamma(d+2)} + \frac{\kappa(\theta-v_0)\E^{-\kappa T}}{T\Gamma(d+2)}\int_0^T t^{d+1}\E^{\kappa t}\D t.
\end{align}
Note in particular that, when $\eta = d = 0$, 
this representation coincides with the expected variance formula in the standard Heston case, 
provided in~\cite[Chapter~11]{GatheralBook}. 
Furthermore, for small~$T>0$, Equation~\eqref{eq:VS} reads
$$
\frac{1}{T}\EE\left(\int_0^TV_t^d\D t\right)
= \frac{v_0T^d}{\Gamma(d+2)} + \eta + \Oo(T^{d+1}),
$$
so that in the rough Heston case the variance swap rate explodes when $d<0$ with a rate of~$T^d$.

\begin{figure}[h]
\centering
\includegraphics[scale=0.5]{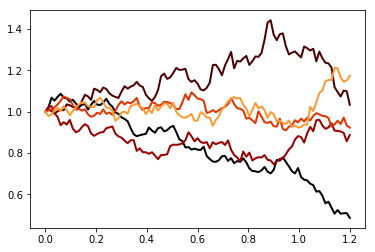}\qquad
\includegraphics[scale=0.5]{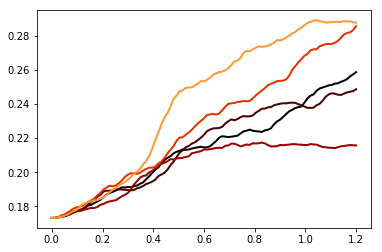}\\
\includegraphics[scale=0.5]{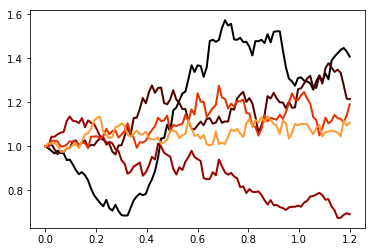}\qquad
\includegraphics[scale=0.5]{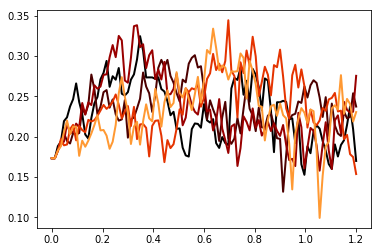}
\caption{We generate five smooth ($d=0.4$, above) and rough ($d=-0.3$, bottom) 
paths for the stock prices (left) and their volatility movements (right). 
Other parameters for the variance process are 
$(\kappa,\theta,\xi,v_0,\eta) = (2.1, 0.05, .2, 0,0.03)$.}
\label{fig:path}
\end{figure}
\subsection{Cumulant generating function}\label{sec:CGF}
Let 
$\mm(u,w,t) \equiv\log\EE(\E^{uX_t + w V_t^d})$, for any $t\geq 0$ and $(u,w)\in \widetilde{\Dd}_t$,
denote the cumulant generating function of the couple~$(X, V^d)$,
where
$\widetilde{\Dd}_t:=\left\{(u,w)\in \RR^2: |\mm(u,w,t)|<\infty\right\}$ is the effective domain of~$\mm$.
The following theorem provides a closed-form expression for it:

\begin{theorem}\label{thm:cgf}
For any $t\geq 0$,
\begin{equation}\label{eq:moment1}
\log\mathbb{E}\left(\E^{uX_t+wV_t^d}\right)
 = w\eta+\frac{u(u-1)\eta t}{2}- B(t)v_0 + A(t),
\qquad\text{for all }(u,w)\in \widetilde{\Dd}_t,
\end{equation}
where the functions $A$ and $B$ satisfy the ordinary differential equations
\begin{equation}\label{eq:ABfracH}
\begin{array}{ll}
 & \displaystyle A'(s)+\kappa\theta B(s)=0,   \\ 
 & \displaystyle B'(s)+\kappa B(s) +\frac{\xi^2}{2} B(s)^2 +\frac{u(u-1)}{2\Gamma(d+1)}s^d+\frac{w}{\Gamma(d)}s^{d-1} =0 ,
\end{array}
\end{equation}
for $0\leq s \leq t$, with the initial conditions $A(0)=B(0)=0$,
and where $\Gamma$ is the usual Gamma function.
\end{theorem}
It is interesting to note that the couple $(X, V^d)$ remains affine in the sense of~\cite{DFS}.
As outlined in Remark~\ref{rem:NonMarkov}, correlation would break the Markovian property of the process,
as well as its affine property.
When $\eta=d=w=0$, the Riccati equations~\eqref{eq:ABfracH} are the same as in the standard (uncorrelated) Heston model.
This in particular implies that, when correlation is null, our model~\eqref{eq:fhm} and the Heston model 
have the same marginals.
The analysis of the asymptotic behaviour of the implied volatility below
will only require the knowledge of the function~$m$ evaluated at $w=0$,
and we shall write $\mm(u,w,t) = \mm(u,t)$ with effective domain~$\Dd_t$.
\begin{proof}
The map $(s,u) \rightarrow \ind_{\{0\leq s\leq u\}} \frac{(s-u)^{d-1}}{\Gamma(d)} V _u$ 
is non-negative on~$[0,t]\times [0,t]$ a.s., so that, by Tonelli's theorem,
$$
\int_0^t (V^d_s-\eta) \D s
 = \int_0^t \int_0^s \frac{(s-u)^{d-1}}{\Gamma(d)} V _u \D u \D s
 = \int_0^t \int_0 ^t \ind_{\{0\leq u\leq s\}} \frac{(s-u)^{d-1}}{\Gamma(d)} V _u \D s \D u
 = \int_0^t \frac{(t-u)^d}{\Gamma(d+1)}V_u \D u.
$$
Since moments of the integrated CIR exist, up to explosion~\cite[Part I, Chapter 6.3]{JYCBook}, 
the Novikov condition then justifies the introduction, for any $u\in\Dd_t$,
of the measure $\widetilde{\mathbb{P}}$ via the Radon-Nikodym derivative 
\begin{equation}\label{eq:measchanfrach}
\left.\frac{\D \widetilde{\mathbb{P}}}{\D\mathbb{P}}\right|_{\Ff_t}
 := \exp\left(u\int_{0}^{t}\sqrt{V_s^d}\D B_s-\frac{u^2}{2}\int_{0}^{t}V_s^d\D s\right).
\end{equation}
The moment generating function can then be written as
\begin{align*}
\mathbb{E}\left[\E^{uX_t+wV_t^d}\right]
 & = \E^{w\eta}\widetilde{\mathbb{E}}\left[
\exp\left(\frac{u(u-1)}{2}\int_0^tV_s^d\D s+w\int _0^t\frac{(t-s)^{d-1}}{\Gamma(d)}V_s \D s\right)\right]\\
 & = \exp\left(w\eta+\frac{u(u-1)\eta t}{2}\right)
\widetilde{\mathbb{E}}\left[\exp\left(\int_0^tV_s\left(\frac{u(u-1)}{2\Gamma(d+1)}(t-s)^d+\frac{w(t-s)^{d-1}}{\Gamma(d)}\right)\D s\right)\right].
\end{align*}
Furthermore, under $\widetilde{\mathbb{P}}$, the process~$V$ satisfies
\begin{equation}\label{eq:dynV}
\D V_t= \kappa\left(\theta-V_t\right)\D t+\xi\sqrt{V_t}\D W_t.
\end{equation}
Set now
$\psi(z,s)
:=\widetilde{\mathbb{E}}\left[\left.\exp\left\{
\int_s^t V_r\left(\frac{u(u-1)}{2\Gamma(d+1)}(t-r)^d+\frac{w(t-r)^{d-1}}{\Gamma(d)}\right)\D r\right\}\right| V_s=z\right]$.
We want $\psi(v_0,0)$ and $\psi$ solves the following PDE using the Feynman-Kac formula:
$$
\partial_s \psi(z,s)
+\kappa(\theta-z) \partial_z \psi(z,s)
+\frac{1}{2}z\xi^2  \partial^2_z \psi(z,s)
+\left(\frac{u(u-1)}{2\Gamma(d+1)}(t-s)^d
+\frac{w}{\Gamma(d)}(t-s)^{d-1} \right) z
\psi(z,s)
=0,
$$
for $0\leq s \leq t$, with terminal condition $\psi(z,t)=1$ for $z\geq0$.
The change of variables $\varphi(z,s)\equiv \psi(z,t-s)$ yields
$$
-\partial_s \varphi(z,s)
+\kappa(\theta-z) \partial_z \varphi(z,s)
+\frac{1}{2}z\xi^2  \partial^2_z \varphi(z,s)
+\left(\frac{u(u-1)}{2\Gamma(d+1)}
s^d
+\frac{w}{\Gamma(d)}s^{d-1} \right) z
\varphi(z,s)
=0,
$$
for $0\leq s \leq t$ and with initial condition $\varphi(z,0)=1$ for all $z\geq0$.
Using the ansatz
$
\varphi(z,s)=\E^{A(s;u,t)-B(s;u,t)z},
$
we find that $A(\cdot;u,t),B(\cdot;u,t)$ solve the Ricatti ODEs in~\eqref{eq:ABfracH}.
\end{proof}

\begin{remark}\label{rem:NonMarkov}
For non-zero correlation, the variance dynamics under $\widetilde{\mathbb{P}}$ in~\eqref{eq:dynV} change to
$$
\D V_t= \kappa\left(\theta-V_t\right)\D t+\rho\xi u \sqrt{V_t V_t^d}\D t+\xi\sqrt{V_t}\D W_t,
$$
which is non-Markovian, hence not amenable to Feynman-Kac techniques,
and is left for future research.
\end{remark}

\begin{remark}\label{rem:Explicitkappa0}
In the case where $\kappa=0$, the ODEs~\eqref{eq:ABfracH} can be solved explicitly~\cite{PZ03}, and 
$$
\mm(u,t)=\frac{u(u-1)\eta t}{2}-\frac{2v_0}{\xi^2}\partial_t \log \left( \sqrt{t}\left[C_1 J_{\frac{1}{d+2}}
\left(\frac{\xi}{d+2}\sqrt{\frac{u(u-1)}{\Gamma(d+1)}}t^{d/2+1}\right) 
+C_2 Y_{\frac{1}{d+2}}
\left(\frac{\xi}{d+2}\sqrt{\frac{u(u-1)}{\Gamma(d+1)}}t^{d/2+1}\right)
\right] \right),
$$
where $J$ and $Y$ are the Bessel functions of respectively the first and the second kind, 
and $C_1,C_2$ are constants determined by the boundary conditions.
Using $J_{1/2}(x)=\sqrt{2/(\pi x)}\sin x$ and $Y_{1/2}(x)=-\sqrt{2/(\pi x)}\cos x$ 
we recover the Heston mgf~\cite{JF10} (with $\kappa=0$) when $d=0$ and $\eta=0$, namely
$\mm(u,t)=\frac{v_0}{\xi}\sqrt{u(u-1)}\tan\left(\xi t \sqrt{u(u-1)}/2 \right)$.
\end{remark}
\begin{figure}[h]
\centering
\includegraphics[scale=0.37]{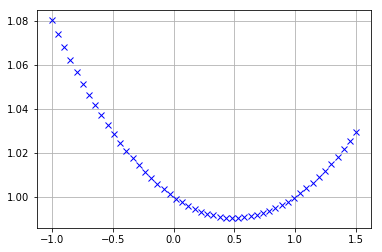}\qquad
\includegraphics[scale=0.37]{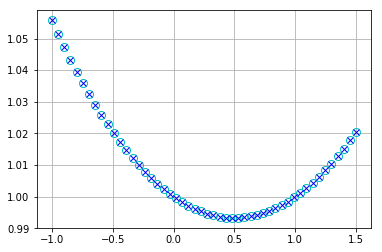}
\caption{
We numerically compute the mgf of~$X_t$ using
the ODE in Theorem~\ref{thm:cgf} with~$t=1$. 
On the left 
$(\kappa, \theta, \xi, v_0, \eta, d) = (2.1, 0.05, 0.2, 0.03, 0.03, -0.3)$.
As a sanity check on the right we compare the numerically solved mgf for 
$d=0$,~$v_0=0.06, \eta=0$ with
the explicit representation of the mgf in a Heston setting (circles),
provided in~\cite[Chapter~2]{GatheralBook}.}
\label{fig:fracHesMGF}
\end{figure}
\section{Large deviations and implied volatility asymptotics}\label{sec:Asymptotics}
\subsection{Large deviations of the log stock price process}
This section gathers the main results of the paper, 
namely a large deviations principle for a rescaled version of the log stock price 
process~$X$ defined in~\eqref{eq:fhm} when time becomes large or small.
We refer the reader to Appendix~\ref{app:LDP} 
for a brief reminder of large deviations and to~\cite{DZ, DS} for thorough treatments.
Before stating (and proving) the main results of the paper, 
introduce the functions $\lambda_\pm^*, \Lambda_{\pm}^*:\RR\to\RR_+$,
$\Lambda_+:[0,1]\to\RR$ and $\Lambda_-:\RR\to\RR$:
\begin{equation}\label{eq:RateFunctions}
\begin{array}{ll}
\Lambda_+(u) \equiv \displaystyle -\frac{\kappa \theta}{\xi(1+d/2)}\sqrt{\frac{u(1-u)}{\Gamma(1+d)}},
& \Lambda_-(u) \equiv \displaystyle \frac{1}{2}u(u-1)\eta,\\
\Lambda_+^*(x)
\displaystyle \equiv \urm^*(x)x+\frac{\kappa \theta}{\xi(1+d/2)}\sqrt{\frac{u^*(x)(1-u^*(x))}{\Gamma(1+d)}},
 & \Lambda_-^*(x)
 :=
\left\{
\begin{array}{ll}
\displaystyle \frac{(x + \eta/2)^2}{2\eta},
\quad & \text{if }x\in [\Lambda_-'(u_-),\Lambda'_-(u_+)],\\
u_+x-\Lambda_-(u_+), \quad  &\text{if }x>\Lambda'_-(u_+),\\
u_-x-\Lambda_-(u_-), \quad & \text{if }x<\Lambda'_-(u_-),
\end{array}
\right.
\\
\lambda_+^*(x)
 \equiv \displaystyle \frac{x^2}{2\eta},
 & \lambda_-^*(x)
 \equiv \displaystyle \frac{\Gamma(2+d)x^2}{2v_0}.
\end{array}
\end{equation}
where the real numbers $u_-,u_+$ are defined in Theorem~\ref{prop:LargeTimeCGF} and
$$
\urm^*(x)
\equiv \frac{1}{2}\left\{1+ \sgn(x) \sqrt{1 - \left[1+\frac{1}{4}\left(\frac{x\xi \sqrt{\Gamma(1+d)}}{\kappa \theta}\right)^2\right]^{-1}} \right\},
$$
with $\sgn(x) := \ind_{\{x\geq 0\}} - \ind_{\{x<0\}}.$
The expressions $\Lambda'_-(u_-)$ and $\Lambda'_-(u_+)$ are shorthand notations
for, respectively $\lim_{u\downarrow u_-}\Lambda'_-(u)$ and $\lim_{u\uparrow u_+}\Lambda'_-(u)$.
Straightforward computations show that the function~$\urm^*$ is smooth on $\RR^*$,
decreasing and concave on $\RR^*_-$ and increasing and concave on $\RR^*_+$, 
with $\urm^*(0) = 1/2$ and
\begin{equation*}
\begin{array}{rll}
\displaystyle \lim_{x\uparrow 0} \urm^{*(n)}(x) & = \displaystyle - \lim_{x\downarrow 0} \urm^{*(n)}(x),
 & \text{if } n \text{ is odd};\\
\displaystyle \lim_{x\uparrow 0} \urm^{*(n)}(x) & = \displaystyle \lim_{x\downarrow 0} \urm^{*(n)}(x) = 0,
 & \text{if } n \text{ is even}.
\end{array}
\end{equation*}
Therefore, $\Lambda^*_+$ is strictly convex and smooth and maps $\RR$ to $\RR_+$.
Straightforward computations also yield that the function $\Lambda_-^*$ is smooth on the real line.
Note that~$\Lambda^*_-$ and~$\Lambda^*_+$ are nothing else than the Fenchel-Legendre transforms
of~$\Lambda_-$ and~$\Lambda_+$.
Likewise, the functions $\lambda_+^*$ and $\lambda_-^*$ are clearly strictly convex on the whole real line.
\begin{theorem}\label{thm:LDP}\ 
\begin{enumerate}[(i)]
\item As $t$ tends to zero, 
\begin{enumerate}[(a)]
\item if $d\in (0,1/2)$, then $(X_t)_{t\geq 0}$ satisfies a LDP with good rate function $\lambda_+^*$ and speed $t^{-1}$;
\item if $d\in (-1/2,0)$, then $(X_t)_{t\geq 0}$ satisfies a LDP with good rate function $\lambda_-^*$ and speed $t^{-(1+d)}$;
\end{enumerate}
\item As $t$ tends to infinity, 
\begin{enumerate}[(a)]
\item if $d\in (0,1/2)$, then $\left(t^{-(1+d/2)}X_t\right)_{t>0}$ satisfies a LDP
with good rate function $\Lambda^*_+$ and speed $t^{-(1+d/2)}$;
\item if $d\in (-1/2,0)$, then $\left(t^{-1}X_t\right)_{t>0}$ satisfies a partial LDP
on $(\Lambda_-'(u_-), \Lambda_-'(u_+))$
with rate function $\Lambda^*_-$ and speed $t^{-1}$,
where $u_-, u_+$ are defined in Proposition~\ref{prop:LargeTimeCGF}.
\end{enumerate}
\end{enumerate}
\end{theorem}
In practice, the partial LDP in case (ii)(b)is enough here since strikes are of the form~$\E^{xt}$,
for large maturity~$t$ and fixed~$x$; furthermore $0\in(\Lambda_-'(u_-), \Lambda_-'(u_+))$, 
so that for fixed small~$x$, and~$t$ large enough, even large strikes can be computed,
see Theorem~\ref{thm:ImpliedVol} (ii)(b) for details.
For the effective domain~$\Dd_t$ of $\mm(\cdot,t)$ defined in Section~\ref{sec:CGF}, let
$\Dd_\infty := \cap_{t>0}\cup_{s\leq t} \Dd_t = \cup_{t>0}\cap_{s\leq t}\Dd_t$,
and $\Dd_0^{(\delta)}$ the effective domain of the pointwise limit $u\mapsto \mm(t^{-\delta}u, t)$
as $t$ tends to zero ($\delta>0$).
The proof of Theorem~\ref{thm:LDP} relies on the study of the limiting behaviour of the cumulant generating function
of (a rescaled version of) the process~$(X_t)_{t\geq 0}$, 
which we state in the following two propositions
(and defer their proofs to Sections~\ref{sec:ProofSmallCGF} and~\ref{sec:ProofLargeCGF}):

\begin{proposition}\label{prop:SmallTimeCGF}
The following hold as~$t$ tends to zero:
\begin{enumerate}[(i)]
\item if $d \in (0,1/2]$, let $\delta = 1$, then $\Dd_0^{(\delta)} = \RR$ and 
$\displaystyle 
\lim_{t\downarrow 0}t \mm\left(\frac{u}{t},t\right) = \eta \frac{u^2}{2}$,
for $u\in \Dd_0^{(\delta)}$;
\item if $d \in [-1/2,0)$, let $\delta = 1+d$, then $\Dd_0^{(\delta)} = \RR$ and 
$\displaystyle 
\lim_{t\downarrow 0}t^{1+d}\mm\left(\frac{u}{t^{1+d}},t\right) = \frac{v_0 u^2}{2\Gamma(2+d)}$,
for $u\in \Dd_0^{(\delta)}$.
\end{enumerate}
\end{proposition}
\begin{proposition}\label{prop:LargeTimeCGF}
As $t$ tends to infinity, we have the following behaviours for the cumulant generating function:
\begin{enumerate}[(i)]
\item if $d \in (0,1/2)$, then $\Dd_\infty = [0,1]$ and 
$\displaystyle 
\lim_{t\uparrow\infty} t^{-(1+d/2)} \mm(u,t) = \Lambda_+(u)$,
for $u\in \Dd_\infty$;
\item if $d \in (-1/2,0)$, then there exist $u_-\leq 0$, $u_+\geq 1$ such that 
$\displaystyle 
\lim_{t\uparrow\infty} t^{-1} \mm(u,t) = \Lambda_-(u)$,
for $u\in \Dd_\infty = [u_-, u_+]$.
\end{enumerate}
\end{proposition}

\begin{remark}\label{rem:Limitmgf}
The limits above are not continuous in $d$ at the origin.
In the case $d=0$ (standard Heston), the pointwise limit of the rescaled cumulant generating function was computed in~\cite{FJ10} and is such that
$\lim_{t\uparrow\infty}t^{-1}\mm(u,t)$ is a smooth convex function on some interval 
$[u_-^0,u_+^0]\supset [0,1]$.
In Proposition~\ref{prop:LargeTimeCGF}, the characterisation of the limiting domain $\Dd_\infty$ is not
fully explicit. However, using comparison principles between the ODEs~\eqref{eq:ABfracH}
and the corresponding ones in the standard (uncorrelated) Heston model, 
it is easy to see that the interval $[u_-, u_+]$ is contained in $[u_-^H, u_+^H]$, 
which is the limiting domain of the rescaled moment generating function 
in the uncorrelated Heston model (see~\cite{FJ10} for details and explicit expressions for~$u_\pm^H$).
\end{remark}

\begin{proof}[Proof of Theorem~\ref{thm:LDP}]
From Proposition~\ref{prop:SmallTimeCGF}, the large deviations principle stated in~(i)(a)-(b)
follows from a direct application of the G\"artner-Ellis theorem (Theorem~\ref{thm:GE}).
Consider now the large-time behaviour, and start with Case~(ii)(a), i.e. $d\in (0,1/2)$.
From Proposition~\ref{prop:LargeTimeCGF}, the function~$\Lambda_+$ is essentially smooth on~$\Dd_\infty$,
but the origin is not in the interior of~$\Dd_\infty$, 
and hence the G\"artner-Ellis theorem does not apply directly.
However, in the one-dimensional case, one can use the refined version in~\cite{Brien},
which relaxes this assumption.
Case (ii)(b) is a direct application of the G\"artner-Ellis theorem on the effective domain
$[u_-,u_+]$.
\end{proof}
\begin{remark}
One could prove the lower bound on the whole real line even when the steepness condition is not satisfied,
by introducing a well chosen time-dependent change of measure, as in~\cite{Bercu} or~\cite{JR13};
however, this requires knowledge of higher orders of the asymptotic behaviour of the cumulant generating function. We leave this for future research.
\end{remark}
\begin{figure}[h]
\centering
\includegraphics[scale=0.37]{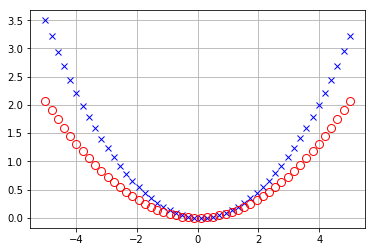}\qquad
\includegraphics[scale=0.37]{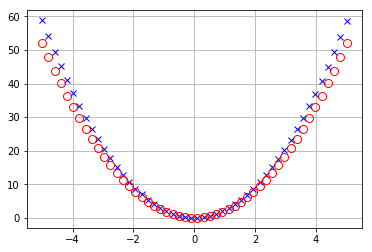}\\
\includegraphics[scale=0.37]{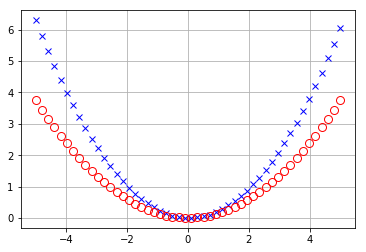}\qquad
\includegraphics[scale=0.37]{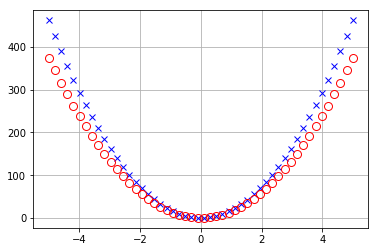}
\caption{We compute the rescaled cgf of~$X_t$ (blue),
and compared it with the small-time limits (red) given in Proposition~\ref{prop:SmallTimeCGF}.
The time parameter is $t=1/10$ on the left column, and $t=10^{-3}$ on the right.
The first row is the rough case where~$d=-0.3$, 
and for the smooth case we have~$d=0.2$ at the bottom. 
$(\kappa,\theta, \xi, v_0, \eta) = (2.1,0.06, 0.2, 0.03, 0.03)$.}
\label{fig:LogMGF}
\end{figure}
\subsection{Implied volatility asymptotics}
We now translate the large deviations principles derived above into asymptotics of the implied volatility.
For any $(x,t)\in \RR\times\RR_+$, let~$\Sigma(x,t)$ denote the implied volatility corresponding to 
European option prices with maturity~$t$ and strike~$\E^{x}$.
The following theorem is proved in Section~\ref{sec:ProofSmallVol} (small time)
and in Section~\ref{sec:ProofLargeVol} (large time).
\begin{theorem}\label{thm:ImpliedVol}\ 
\begin{enumerate}[(i)]
\item As $t$ tends to zero,
\begin{enumerate}[(i)]
\item If $d\in (0,1/2]$, then for any $x\ne 0$, $\Sigma(x,t)^2$ converges to $\eta$;
\item if $d\in [-1/2,0)$, then for any $x\ne 0$, $t^{-d}\Sigma(x,t)^2$ 
converges to $v_0 / \Gamma(d+2)$;
\end{enumerate}
\item as $t$ tends to infinity, the implied volatility behaves as follows:
\begin{enumerate}[(a)]
\item if $d \in (0,1/2)$, then 
$\lim\limits_{t\uparrow\infty}t^{-d/2}\Sigma(x t^{1+d/2})^2
 = 2\left(2\Lambda_+^*(x) - x + \sqrt{\Lambda_+^*(x)\left(\Lambda_+^*(x)-x\right)} \right)$,
for all $x\in\RR$;
 \item if $d \in (-1/2,0)$, then $\Dd_\infty \supset [0,1]$ and
$\lim_{t\uparrow\infty}\Sigma(xt,t)^2 = \eta$, 
for all $x\in (\Lambda'_-(u_-),\Lambda'_-(u_+))$;
\end{enumerate}
\end{enumerate}
\end{theorem}
For large maturities, as $d$ approaches zero from above, 
we recover the standard Heston implied volatility asymptotics derived in~\cite{FJ10};
in the long memory case ($d>0$), the steepness of the implied volatility is more pronounced for very large maturity due to the $t^{d/2}$ factor than the standard Heston model.
The small-maturity case is especially interesting.
In the case of long memory ($d>0$), the implied volatility converges to a constant at the same speed as Black-Scholes (or for that matter as any diffusion with continuous paths).
In the short-memory regime, the implied volatility blows up at the speed $t^{d/2}$.
It is well documented~\cite{BCC97} that classical stochastic volatility models (driven by standard Brownian motions) 
are not able to capture the observed steepness of the observed implied volatility smile.
Several authors~\cite{MijTankov, Tankov} have suggested the addition of jumps in order to fit this steepness.
Assume that the martingale stock price is given by $S = \E^{X}$, where $X$ is a L\'evy process 
with L\'evy measure supported on the whole real line.
As the maturity~$t$ tends to zero, the corresponding implied volatility behaves as
$\lim_{t\downarrow 0}2 t\log(t) \Sigma(x,t)^2= -x^2$,
for all $x\ne 0$.
The speed of divergence $t\log(t)$ is then to be compared with the $t^d$ ($d\in (-1/2,0)$) 
in the fractional framework above:
clearly, for small enough~$t$, the L\'evy implied volatility blows up much faster.
Note further that, in the limit as maturity tends to zero, the latter does depend on the strike,
but the fractional implied volatility does not.
Our results therefore show that fractional Brownian motion (driving the instantaneous volatility) are an alternative to jumps, and provide smiles steeper than standard stochastic volatility, but less steep than L\'evy models. 
They further have the advantage of bypassing the hedging issues in jump-based models.
One should also compare this explosion rate to that of the implied volatility 
in the standard Heston model, where the instantaneous variance is started from an initial distribution.
In~\cite{JR13}, the authors chose a non-central chi-squared initial distribution and showed that,
under some appropriate rescaling, the implied volatility blows up at speed~$t^{-1/2}$.
Jacquier and Shi~\cite{Shi17}, also inspired by~\cite{Mechkov}, pushed this analysis further
by studying the impact of the random initial data on the short-time explosion of the smile.


\section{Numerics}\label{sec:numerics}
In this section we provide numerical examples describing the behaviour of our model,
with parameters $(\kappa, \theta, \xi)= (2.1, 0.06, 0.2)$. 
In Figure~\ref{fig:VolSurf} we provide a comparison of volatility surfaces 
generated by fractional and standard Heston models. 
We observe that in the case where the Hurst parameter is less than~$\frac{1}{2}$ ($d<0$),
a larger value of~$v_0$ pushes up the small-time volatility smile,
resonating our small-time analysis in Theorem~\ref{thm:ImpliedVol}.
Also notice that with similar parameter choices 
the at-the-money implied volatilities are higher in the fractional case. 
This is further confirmed by Figure~\ref{fig:totvar} representing the term structure of the at-the-money total variance.

\begin{figure}[h]
\centering
\includegraphics[scale=0.37]{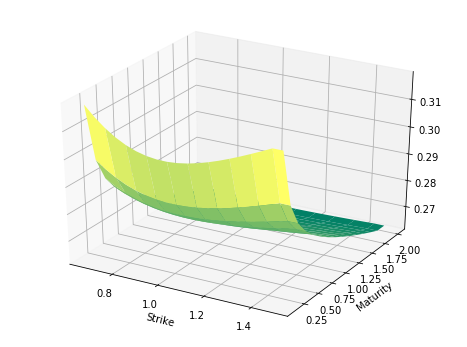}\qquad
\includegraphics[scale=0.37]{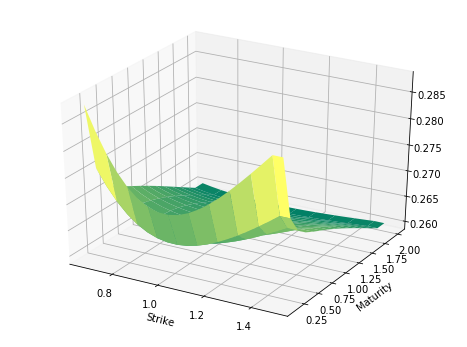}\\
\includegraphics[scale=0.37]{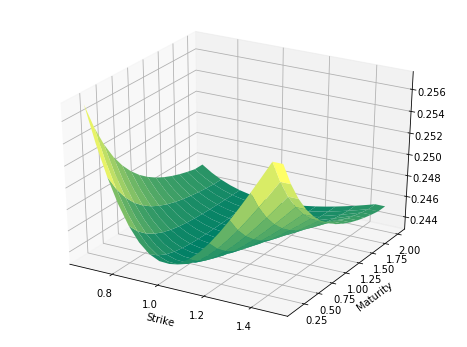}\qquad
\caption{
Comparison of volatility surfaces.
The first row represents the fractional case with
$(v_0, \eta, d) = (0.03, 0.02, -0.3)$ on the left and
$(v_0, \eta, d) = (0.02, 0.02, -0.3)$ on the right.
The bottom plot is the standard Heston case ($d=0$) with $v_0 = 0.06$.
}
\label{fig:VolSurf}
\end{figure}
\begin{figure}[h]
\centering
\includegraphics[scale=0.37]{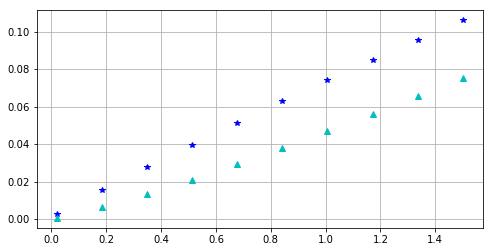}
\caption{
Term structure of the at-the-money total variance 
in fractional (blue, $(v_0, \eta, d) = (0.03,0.02, -0.3)$) and standard Heston (cyan, $v_0 = 0.06$).
}
\label{fig:totvar}
\end{figure}
\begin{figure}[h]
\centering
\includegraphics[scale = .37]{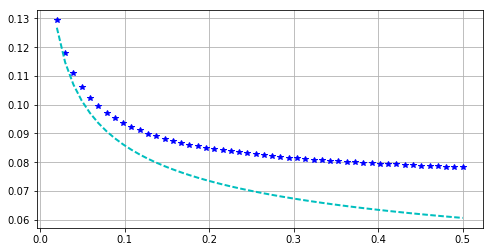}
\caption{
Term structure of the expected annualised variance (blue) in~\eqref{eq:VS} where
$(v_0,\eta, d) = (0.03, 0.02, -0.3)$.
Cyan dashed lines represent its small-time approximation~$\left( \frac{v_0T^d}{\Gamma(d+2)} + \eta\right)$.
}
\end{figure}

\section{Proofs}\label{sec:Proofs}
\subsection{Proof of Proposition~\ref{prop:SmallTimeCGF}}\label{sec:ProofSmallCGF}
We separate the $\kappa=0$ and $\kappa > 0$ in the proof.

\subsubsection{The $\kappa=0$ case}
One could use the explicit knowledge of the moment generating function in Remark~\ref{rem:Explicitkappa0}
to compute its limiting behaviour.
However, we follow a different route here, which we can adapt later to the $\kappa>0$ case.
The function $B$ is the solution of the ODE
\begin{equation}\label{eq: odeB}
B'(t)  = \displaystyle -\frac{\xi ^2}{2}B^2(t)+\frac{u(1-u)}{2\Gamma(d+1)}t^d,
\end{equation}
with boundary condition $B(0) =0$.
Let $\zeta:=\frac{u(1-u)}{2\Gamma(d+2)}$, and consider the ansatz
$B(t)=\sum \limits_{i\geq 1} \alpha_i \zeta^i t^{id+2i-1}$, 
for some sequence of real numbers $(\alpha)_{i\geq 1}$. 
Therefore
\begin{align*}
B'(t)  +\frac{\xi ^2}{2}B^2(t)-\frac{u(1-u)}{2\Gamma(d+1)}t^d
 & = \sum \limits_{i\geq 1}\alpha_i \zeta^i (id+2i-1) t^{id+2i-2}
 + \frac{\xi ^2}{2}\sum \limits_{i,j\geq1} \alpha_i \alpha_j \zeta^{i+j} t^{(i+j)d+2(i+j)-2}
 - \frac{u(1-u)}{2\Gamma(d+1)}t^d \\
 & = \left(\alpha_1 (1+d) \zeta -\frac{u(1-u)}{2\Gamma(d+1)}\right) t^d +\sum \limits_{i\geq 2} \left[\alpha_i  (id+2i-1)+\frac{\xi ^2}{2}\sum \limits_{k=1}^{i-1} \alpha_k \alpha_{i-k} \right]  \zeta^i t^{id+2i-2}.
\end{align*}
This polynomial (in~$t$) is null everywhere if and only if 
$\alpha_1 = 1$ and 
$\alpha_i = \displaystyle -\frac{\xi ^2}{2(i(d+2)-1)}\sum \limits_{k=1}^{i-1} \alpha_k \alpha_{i-k}$,
for $i\geq 2$.
We now make the above derivation rigorous.
Consider the map $f:\RR_+\to\RR$ defined by
$f(t) := t^{-1}\sum_{i\geq 1}\alpha_i \left(\zeta t^{d+2}\right)^i$,
for $t\geq 0$.
For $i$ large enough, $\left|\frac{\xi ^2}{2(i(d+2)-1)}\right|<1$, and hence $|\alpha_i|<|\gamma_i|$, 
where the sequence $(\gamma)_{i\geq 1}$ is defined by
$\gamma_1 = 1$
and
$\gamma_i =\sum_{k=1}^{i-1} \gamma_k \gamma_{i-k}$, for $i\geq 2$.
Tedious but straightforward computations show that the power series
$m(x)\equiv\sum_{i\geq 1} \gamma_i x^i$ has a radius of convergence $1/4$,
and $m'(x)=\sum_{i\geq 1} \gamma_i i x^{i-1}$ for all $x\in (-1/4,1/4)$.
Therefore, the power series $g(x)\equiv\sum_{i\geq 1} \alpha_i x^i$ 
has a radius of convergence greater than~$1/4$ and 
$g'(x)=\sum_{i\geq 1} \alpha_i i x^{i-1}$ for all $x\in (-1/4,1/4)$.
For $|\zeta x^{d+2}|$ small enough, 
$$
\frac{\D}{\D x} \sum_{i\geq 1} \alpha_i \left(\zeta x^{d+2}\right)^i
 = (d+2)\zeta  x^{d+1}\sum_{i\geq 1} \alpha_i i (\zeta x^{d+2})^{i-1},
$$
so that $f'(x)= \sum_{i\geq 1} \alpha_i \zeta^i (id+2i-1) t^{id+2i-2}$. 
Moreover, using Tonelli's theorem, for $|\zeta t^{d+2}|$ small enough,
$$
\sum \limits_{i\geq1, j\geq 1}|\alpha_i | |\alpha_j ||\zeta |^{i+j} t^{(i+j)d+2(i+j)-2}
 = \left(\sum \limits_{i\geq1}|\alpha_i | |\zeta |^{i} t^{id+2i-1}\right)^2
$$
is finite.
A direct application of Fubini's theorem then yields
$f^2 (t)=\sum_{i, j\geq 1} \alpha_i \alpha_j  \zeta^{i+j} t^{(i+j)d+2(i+j)-2}$,
and hence~$f$ is the solution of~\ref{eq: odeB} for small~$|\zeta t^{d+2}|$.

We first prove the proposition in the case $d\in (0,1/2)$.
The function $B$, solution to~\eqref{eq: odeB}, satisfies:
\begin{equation*}
\begin{array}{rl}
\displaystyle t^{1-d}B\left(\frac{u}{t},t\right)
&=\displaystyle t^{-d}\sum_{i\geq 1} \alpha_i \left( \frac{u t^{d+1}}{2\Gamma(2+d)} -
 \frac{u^2 t^d}{2\Gamma(2+d)}\right)^i
=\displaystyle  t^{-d}\sum_{i\geq 1} \alpha_i t^{id}\left( \frac{u t}{2\Gamma(2+d)}
 - \frac{u^2}{2\Gamma(2+d)}\right)^i \\
&=\displaystyle  \sum_{i\geq 1} \alpha_i t^{(i-1)d}\left( \frac{u t}{2\Gamma(2+d)}
 - \frac{u^2}{2\Gamma(2+d)}\right)^i \\
&=\displaystyle \frac{u t}{2\Gamma(2+d)}
 - \frac{u^2}{2\Gamma(2+d)} + \sum_{i\geq 1} \alpha_{i+1} t^{id}\left(\frac{u t}{2\Gamma(2+d)}
 - \frac{u^2}{2\Gamma(2+d)}\right)^{i+1}.
\end{array}
\end{equation*}
Moreover, for $ t $ small enough,
$$
\left| \sum_{i\geq 1} \alpha_{i+1} t^{id}
\left( \frac{ut}{2\Gamma(2+d)} - \frac{u^2}{2\Gamma(2+d)}\right)^{i+1} \right|
 \leq t^{d/2} \left|  \left( \frac{ut}{2\Gamma(2+d)} - \frac{u^2}{2\Gamma(2+d)}\right)
 \sum_{i\geq 1} \alpha_{i+1}\left( \frac{u t^{1+d/2}}{2\Gamma(2+d)} -\frac{u^2 t^{d/2}}{2\Gamma(2+d)} \right)^{i} \right|.
$$
Since for $t$ small enough, 
$\left|\frac{u t^{1+d/2}}{2\Gamma(2+d)}-\frac{u^2 t^{d/2}}{2\Gamma(2+d)}\right|<\frac{1}{4}$, 
then 
$\sum_{i\geq 1} \alpha_{i+1}\left( \frac{u t^{1+d/2}}{2\Gamma(2+d)}-\frac{u^2 t^{d/2}}{2\Gamma(2+d)} \right)^i$ is well defined. 
Hence the left-hand side $\sum_{i\geq 1} \alpha_{i+1} t^{id}\left( \frac{ut}{2\Gamma(2+d)}-\frac{u^2}{2\Gamma(2+d)}\right)^{i+1}$ converges to zero as $t$ tends to zero, 
and so does $t^{1-d}B(u/t,t)$ to $-\frac{u^2}{2\Gamma(2+d)}$.

In the case $d\in [-1/2,0)$, the function $B$, solution to~\eqref{eq: odeB}, satisfies
\begin{equation*}
\begin{array}{rl}
\displaystyle t^{1+d}B(u/t^{1+d},t)
&=\displaystyle t^{d}\sum_{i\geq 1} \alpha_i \left( \frac{u t}{2\Gamma(2+d)}-\frac{u^2 t^{-d}}{2\Gamma(2+d)} \right)^i \\
&=\displaystyle \frac{u t^{1+d}}{2\Gamma(2+d)} -\frac{u^2}{2\Gamma(2+d)}
 + t^{d} \sum_{i\geq 1} \alpha_{i+1} \left( \frac{u t}{2\Gamma(2+d)}
 - \frac{u^2 t^{-d}}{2\Gamma(2+d)}\right)^{i+1} \\
&=\displaystyle \frac{u t^{1+d}}{2\Gamma(2+d)} -\frac{u^2}{2\Gamma(2+d)}
 + \left( \frac{u t^{1+d}}{2\Gamma(2+d)} -\frac{u^2}{2\Gamma(2+d)}\right)
\sum_{i\geq 1} \alpha_{i+1} \left( \frac{u t}{2\Gamma(2+d)} - \frac{u^2 t^{-d}}{2\Gamma(2+d)}\right)^{i} \end{array}
\end{equation*}
Moreover, for $t$ small enough, 
$\left| \sum_{i\geq 1} \alpha_{i+1} \left( \frac{ut}{2\Gamma(2+d)}
 - \frac{u^2 t^{-d}}{2\Gamma(2+d)}\right)^{i}\right|
 \leq t^{-d/2} \left|\sum_{i\geq 1} \alpha_{i+1} \left(\frac{u t^{1+d/2}}{2\Gamma(2+d)}
 - \frac{u^2 t^{-d/2}}{2\Gamma(2+d)}\right)^{i}\right|$. 
Since the series 
$ \sum_{i\geq 1} \alpha_{i+1} \left( \frac{u t^{1+d/2}}{2\Gamma(2+d)}
 -\frac{u^2 t^{-d/2}}{2\Gamma(2+d)}\right)^{i}$ 
is well defined for $t$ small enough, 
$\sum_{i\geq 1} \alpha_{i+1} \left( \frac{u t}{2\Gamma(2+d)} - \frac{u^2 t^{-d}}{2\Gamma(2+d)}\right)^{i}$ tends to zero as $t$ tends to zero, and so does
$t^{1+d}B(u/t^{1+d},t)$ tends to $-\frac{u^2}{2\Gamma(2+d)}$,
which proves the proposition.

\subsubsection{The $\kappa>0$ case}
Similarly to the $\kappa=0$ case, plugging the ansatz
$\displaystyle B(t) :=\sum_{i,j\geq 1}\beta_{i,j}t^{id+j}$ into~\eqref{eq:ABfracH} yields
$$
\left(\beta_{1,1}(d+1) + \frac{u(u-1)}{2\Gamma(d+1)}\right)t^d 
+ \frac{\xi^2}{2}\left(\sum_{i,j\geq 1}\beta_{i,j}t^{id+j}\right)^2
+\kappa\left(\sum_{i,j\geq 1}\beta_{i,j}t^{id+j}\right) + \sum_{(i,j)\ne (1,1)}\beta_{i,j}(id+j)t^{id+j-1} \equiv 0.
$$
The following statements hold, and are related to the case $\kappa=0$ above:
\begin{enumerate}
\item $\beta_{i,j} = 0$ for any~$i\geq 2$ and for~$1\leq j \leq 2i-2$;
\item $\beta_{i,2i-1} = \alpha_i\zeta^i$ for any~$i\geq 1$.
\end{enumerate}
Combining terms of the same order yields
$\beta_{1,1} = \alpha_1\zeta$. 
Moreover, for any~$(i,j)\in\NN_+^*\times\left(\NN_+^*\setminus\{1\}\right)$,
$$
\left(\beta_{i,j}(id+j) + \kappa \beta_{i,j-1} + \frac{\xi^2}{2}\left(\sum_{p=1}^{i-1}\sum_{q=1}^{j-2}\beta_{p,q}\beta_{i-p,j-1-q}\right)\right)t^{id+j-1}\equiv 0.
$$
Then the two statements above can be easily verified by induction,
from which
$\displaystyle B(t) = \sum_{i\geq 1}\sum_{j\geq 2i-1}\beta_{i,j}t^{id+j}$, with 
\begin{equation}\label{eq:beta}
\beta_{i,j} = \left[-\kappa\beta_{i,j-1} - \frac{\xi^2}{2}\left(\sum_{p=1}^{i-1}\sum_{q=2p-1}^{j-2i+2p}\beta_{p,q}\beta_{i-p,j-1-q}\right)\right]\frac{1}{id+j}.
\end{equation}
We show that the series is absolutely convergent.
Notice that
\begin{equation}\label{eq:summation}
\sum_{i\geq 1}\sum_{j\geq 2i-1}|\beta_{i,j}|t^{id+j}
= \sum_{i\geq 1}t^{id+2i-1}\left(\sum_{j\geq 2i-1}|\beta_{i,j}|t^{j-2i+1}\right)
= \sum_{i\geq 1}|\beta_{i,2i-1}|t^{id+2i-1}\left(1+\sum_{j=2i}^\infty\left|\frac{\beta_{i,j}}{\beta_{i,2i-1}}\right|t^{j-2i+1}\right).
\end{equation}
Following Lemma~\ref{lem:beta}, 
$1 + \sum_{j=2i}^\infty\left|\frac{\beta_{i,j}}{\beta_{i,2i-1}}\right|t^{j-2i+1}
< 2^{i-1}\E^{i\kappa t}$ holds for any~$i \geq 1$,
and~\eqref{eq:summation} then reads
\begin{equation}\label{eq:AbsSum}
\sum_{i\geq 1}\sum_{j\geq 2i-1}|\beta_{i,j}|t^{id+j}
<t^{-1}\sum_{i\geq 1}2^{i-1}|\alpha_i|\left(|\zeta|t^{d+2}\E^{\kappa t}\right)^i.
\end{equation}
We already show that~$|\alpha_i| < \gamma_i$ for~$i$ large enough,
then the series $(\sum_{i\geq 1}2^{i-1}|\alpha_i| x^i)$ 
has a radius of convergence no less than~$1/8$.
Since~$(d+2)$ is strictly positive, then~\eqref{eq:AbsSum} implies that the series
$B(t) =  \sum_{i\geq 1}\sum_{j\geq 2i-1}\beta_{i,j}t^{id+j}$ is absolutely convergent for small~$t$
such that $|\zeta|t^{d+2}\E^{\kappa t}<1/8$ holds. 
Therefore the previous ansatz~$B(t)$ is the solution to the Riccati equation for small~$t$.
The rest of the proof is essentially the same as the case where~$\kappa=0$, 
and we therefore omit the details.
\begin{lemma}\label{lem:beta}
For any~$i\geq 1$,
the following estimation for~$\beta_{i,j}$ defined in~\eqref{eq:beta}
holds for any~$j\geq 2i$:
\begin{equation*}
|\beta_{i,j}| < \frac{2^{i-1}(i\kappa)^{j-2i+1}}{(j-2i+1)!}|\beta_{i,2i-1}|.
\end{equation*}
\end{lemma}
\begin{proof}
We prove by induction.
In the case where~$i=1$, direct computations yield
\begin{equation*}
\left|\frac{\beta_{1,j}}{\beta_{1,1}} \right| = \frac{\Gamma(d+2)}{\Gamma(d+1+j)}\kappa^{j-1}
\leq \frac{\kappa^{j-1}}{(j-1)!},\quad\text{for any }j\geq 1.
\end{equation*}
Assume that the upper bound holds for~$|\beta_{p,q}|$ for any~$p\leq i-1$.
Then we have
\begin{equation*}
|\beta_{p,q}| \leq \frac{2^{p-1}(p\kappa)^{q-2p+1}|\beta_{p,2p-1}|}{(q-2p+1)!},\qquad
|\beta_{i-p,j-1-q}| \leq \frac{2^{i-p-1}((i-p)\kappa)^{j-q-2i+2p}|\beta_{i-p,2(i-p)-1}|}{(j-q-2i+2p)!}.
\end{equation*}
As a result, for any fixed~$1\leq p\leq i-1$,
\begin{align*}
\sum_{q=2p-1}^{j-2i+2p}|\beta_{p,q}\beta{i-p,j-1-q}|
&\leq \sum_{q=2p-1}^{j-2i+2p}\frac{2^{i-2}(p\kappa)^{q-2p+1}((i-p)\kappa)^{j-q-2i+2p}|\beta_{p,2p-1}\beta_{i-p,2(i-p)-1}|}{(q-2p+1)!(j-q-2i+2p)!}\\
&= \frac{2^{i-2}|\beta_{p,2p-1}\beta_{i-p,2(i-p)-1}|}{(j-2i+1)!}\sum_{q=2p-1}^{j-2i+2p}\frac{(j-2i+1)!(p\kappa)^{q-2p+1}((i-p)\kappa)^{j-q-2i+2p}}{(q-2p+1)!(j-q-2i+2p)!}\\
&= \frac{2^{i-2}(i\kappa)^{j-2i+1}}{(j-2i+1)!}|\beta_{p,2p-1}\beta_{i-p,2(i-p)-1}|.
\end{align*}
Plug it into~\eqref{eq:beta}, and notice that 
$|\beta_{i,2i-1}| = \frac{\xi^2}{2(id+2i-1)}\sum_{k=1}^{i-1}|\beta_{k,2k-1}\beta_{i-k,2(i-k)-1}|$
which follows from the fact that 
all the terms $\beta_{k,2k-1}\beta_{i-k,2(i-k)-1}$ have the same sign for any~$k$,
then
\begin{align*}
|\beta_{i,j}|
&\leq \frac{\kappa |\beta_{i,j-1}|}{id+j} 
+ \frac{\xi^2}{2(id+j)}\sum_{p=1}^{i-1}\sum_{q=2p-1}^{j-2i+2p}|\beta_{p,q}\beta_{i-p,j-1-q}|\\
&\leq \frac{\kappa |\beta_{i,j-1}|}{id+j} 
+ \frac{\xi^2(id+2i-1)2^{i-2}(i\kappa)^{j-2i+1}}{2(id+j)(id+2i-1)(j-2i+1)!}
\sum_{p=1}^{i-1}|\beta_{p,2p-1}\beta_{i-p,2(i-p)-1}|\\
&\leq \frac{\kappa |\beta_{i,j-1}|}{id+j} 
+ \frac{2^{i-2}(i\kappa)^{j-2i+1}(id+2i-1)}{(j-2i+1)!(id+j)}|\beta_{i,2i-1}|.
\end{align*}
Iterating the inequality above, we finally obtain
\begin{align*}
|\beta_{i,j}|
&\leq \frac{\kappa^{j-2i+1}|\beta_{i,2i-1}|}{\prod_{k=2i}^j(id+k)}
+ 2^{i-2}\kappa^{j-2i+1}|\beta_{i,2i-1}|\sum_{k=2i}^{j}\frac{i^{k-2i+1}(id+2i-1)}{(k-2i+1)!\prod_{m=k}^j(id+m)}\\
&\leq \frac{2^{i-2}\kappa^{j-2i+1}|\beta_{i,2i-1}|}{(j-2i+1)!}
\sum_{k=2i-1}^j i^{k-2i+1}
= \frac{2^{i-2}\kappa^{j-2i+1}|\beta_{i,2i-1}|(i^{j-2i+2}-1)}{(j-2i+1)!(i-1)}\\
&<\frac{2^{i-2}(i\kappa)^{j-2i+1}|\beta_{i,2i-1}|i}{(j-2i+1)!(i-1)}
\leq \frac{2^{i-1}(i\kappa)^{j-2i+1}}{(j-2i+1)!}|\beta_{i,2i-1}|.
\end{align*}
\end{proof}

\begin{remark}
In the case $d=0,\eta=0$ and $\kappa=0 $, the cgf  $\mm(u,t)$ corresponds to the standard Heston model with $\rho=0$ and mean reversion speed $\kappa=0$.
Here it is well known~\cite{JF10} that $\lim_{t\downarrow0}t\mm(u/t,t)=v_0u/(\xi \cot(\xi u/2))$.
Using the series expansion solution to ~\eqref{eq: odeB} above we find that 
$\lim_{t\downarrow0}t \mm(u/t,t)=v_0\sum_{i\geq 1}(-1)^{i+1}\alpha_i (u^2/2)^i$.
Explicitly computing the first few terms we find that $\sum_{i\geq 1}(-1)^{i+1}\alpha_i (u^2/2)^i=u^2/2+\xi^2u^4/24+\xi^4u^6/240+\mathcal{O}(\xi^6 u^8)$, which is in exact agreement with a Taylor expansion of $u/(\xi \cot(\xi u/2))$ for small $\xi^2 u$.
\end{remark}

\subsection{Proof of Proposition~\ref{prop:LargeTimeCGF}}\label{sec:ProofLargeCGF}
We start with the case $d \in (0,1/2)$.
Let $B$ be the solution to the ordinary differential equation~\eqref{eq:ABfracH},
and $\alpha:=d/2$, 
$\beta_0 := \frac{1}{\xi}\sqrt{\frac{u(1-u)}{\Gamma(d+1)}}$,
$\beta_1 := - \frac{\kappa}{\xi^2}$.
The function $f(t)\equiv B(t)-\beta_0 t^\alpha-\beta_1$  satisfies
$
f'(t) = -\frac{1}{2}\xi^2f(t)^2-\xi^2 \beta_0 t^{\alpha} f(t)+\frac{\kappa^2}{2\xi^2}-\beta_0\alpha t^{\alpha-1}$,
for $t>0$.
Define now $\psi_-, \psi_+: \RR_+\to\RR$ by
\begin{equation}\label{eq:PhiPsi}
\psi_\pm(t) \equiv -\frac{\kappa}{\xi^2} \pm \frac{1}{\xi}\sqrt{\frac{\kappa^2}{\xi^2}+\frac{u(1-u)t^d}{\Gamma(d+1)}}.
\end{equation}
We now claim that for any $t>0$, the following inequalities hold:
\begin{equation}\label{eq:Bcontrol}
\psi_-(t) \leq B(t) \leq \psi_+(t).
\end{equation}
Note that 
\begin{equation}\label{eq:Psipm}
\psi_{\pm}'(t) = \pm\frac{\alpha \beta_0^2 t^{2\alpha -1}}{\sqrt{\beta_1^2+\beta_0^2t^{2\alpha}}},
\end{equation}
which implies that $\psi_-'(t) \le 0\le \psi_+'(t)$, for all $t>0$, and 
\begin{align*}
& \lim\limits_{t \downarrow 0}\psi_-'(t) = -\infty, \quad 
B'(0)=0, \quad 
\lim\limits_{t \downarrow 0}\psi_+'(t)=+\infty,\\
& \psi_-(0) = -2\kappa/\xi^2\le B(0) \leq \psi_+(0)=0.
\end{align*}
Furthermore, note that 
\begin{equation}\label{eq:BPrice}
B'(t)=-\frac{\xi^2}{2}(B(t) - \psi_+(t))(B(t) - \psi_-(t)).
\end{equation}
If $\tau := \sup\{t>0: \psi_-(s)\leq B(s)\leq\psi_+(s), \text{ for all }0\leq s\leq t\}$ is finite, then
from monotonicity and continuity of~$B$ on~$[0,\tau]$ we have
$B(\tau) = \psi_+(\tau)$ and $B'(\tau) = 0$.
As a result, $B''(\tau) = \xi^2(\psi_+(\tau) - \psi_-(\tau))\psi_+'(\tau)/2>0$, 
implying that there exists $\tau^*>\tau$ such that~$B'(t)>B'(\tau) = 0$ for $t\in(\tau,\tau^*]$,
hence $B(t)\in(\psi_-(t), \psi_+(t))$ for $t\in(\tau,\tau^*]$, contradicting the finiteness of~$\tau$,
and therefore $\tau = \infty$.
We now prove that $\lim\limits_{t \uparrow\infty} t^{-\alpha}B(t)=\beta_0$.
Define the functions $\phi_+, \phi_-:(t^*,+\infty)\to\RR$ by
$$
\phi_\pm(t) =  -\beta_0 t^\alpha \pm 
\left\{\beta_0^2 t^{2\alpha}+\frac{1}{\xi^2}\left(\frac{\kappa^2}{\xi^2}-2\beta_0\alpha t^{\alpha-1}\right)\right\}^{1/2},
$$
where $t^*\geq 0$ is large enough so that~$\phi_\pm(t)$ exists as a real number ($\alpha-1<0$)
for~$t>t^*$.
For large~$t$, we have
\begin{equation}\label{eq:Really}
\phi_+(t)=\frac{\kappa^2}{2\xi^4 \beta_0}t^{-\alpha} + \mathcal{O}(t^{-3\alpha})
\qquad\text{and}\qquad
\phi_-(t)=-2\beta_0 t^\alpha + \mathcal{O}(t^{-\alpha}).
\end{equation}
For large~$t$, $\phi_+'(t)<0$ and $\phi_-'(t)<0$.
Since $f'(t)=-\frac{\xi^2}{2}(f(t) - \phi_+(t))(f(t)-\phi_-(t))$, 
two cases can occur:
\begin{itemize}
\item there exists $t_0>0$ such that $\phi_+(t_0)\leq f(t_0)$;
\item for all $t>t^*$, $\phi_+(t)>f(t)$;
\end{itemize}
In the first case, the comparison principle implies that for all $t>t_0$, $\phi_+(t)\leq f(t)$ 
(similarly to the proof of~\eqref{eq:Bcontrol}).
Therefore~\eqref{eq:Bcontrol} and 
the definition of~$f$ yield
$\beta_0 t^{\alpha}+\beta_1 + \phi_+(t) \le B(t) \le \psi_+(t)$,
which proves the result.
In the second case, let us prove first that there exists $t_1>0$ such that  $f(t_1)>\phi_-(t_1)$:
if for all $t$ such that $\phi_-$ is well defined we have $f(t)\leq \phi_-(t)$, then
$$
\psi_-(t)\le B(t)=f(t)+\beta_1+\beta_0 t^{\alpha} \le \phi_-(t)+\beta_1+\beta_0 t^{\alpha},
$$
which implies $B(t)\sim -\beta_0 t^{\alpha}$, which contradicts the positivity of~$B'$ 
(see~\eqref{eq:Bcontrol} and~\eqref{eq:BPrice}).
Thus there exists $t_1>0$ such that $f(t_1)>\phi_-(t_1)$.
Therefore, since, for all $t>t^*$,
$$
\phi_-'(t) = -\beta_0\alpha t^{\alpha-1}
- \left(\beta_0^2\alpha t^{2\alpha-1} 
+ \frac{\beta_0\alpha(1-\alpha)}{\xi^2 t^{2-\alpha}}\right)\left[\beta_0^2 t^{2\alpha}+\frac{1}{\xi^2}\left(\frac{\kappa^2}{\xi^2}- 2 \beta_0\alpha t^{\alpha-1}\right)\right]^{-1/2} <0,
$$
the comparison principle implies that 
$\phi_-(t)\leq f(t)\leq \phi_+(t)$ for all $t>t_1$.
Therefore~$f$ is non-decreasing on $(t_1, \infty)$;
being bounded by $\phi_+$, which tends to zero at infinity, it converges to a constant, and 
$\lim\limits_{t \uparrow\infty} t^{-\alpha}B(t)=\beta_0$.

We now prove that the effective domain $\Dd_t$ converges to $[0,1]$ as $t$ tends to infinity.
For any $u \in \RR\setminus [0,1]$,
$$
B'(t)=-\frac{\xi^2}{2}\left( B^2 (t) + \frac{2\kappa B(t)}{\xi^2}\right) +\frac{u(1-u)}{2\Gamma(1+d)}t^d
 = -\frac{\xi^2}{2}\left(B(t) + \frac{\kappa}{\xi^2}\right)^2+\frac{\kappa^2}{2\xi^2}  +\frac{u(1-u)}{2\Gamma(1+d)}t^d.
$$
Therefore, $B'(t)\leq \frac{\kappa^2}{2\xi^2} +\frac{u(1-u)}{2\Gamma(1+d)}t^d$, so that
$B(t) \leq \kappa^2 t/2\xi^2  +\frac{u(1-u)}{2\Gamma(d+2)}t^{1+d}$, and hence
$\lim\limits_{t \uparrow +\infty} t^{-d/2}B(t)=-\infty$.
Since $A(t) = -\kappa\theta \int_{0}^{t} B(s) \D s$, 
and since $[0,1]$ is always in $\Dd_t$, for any $t\geq 0$ (since the process is a martingale),
Part~(i) of the proposition follows from Theorem~\ref{thm:cgf}.\\

We now move on to part~(ii) of the proposition, when $d \in (-1/2,0)$.
Let us first prove that $\lim\limits_{t \uparrow\infty} B(t)=0$.
Consider the functions $\psi_-, \psi_+$ as defined in~\eqref{eq:PhiPsi}. 
As~$t$ tends to zero, $\psi_+$ diverges to $+\infty$ and $\psi_-$ to $-\infty$, so that
$\lim\limits_{t \downarrow0} \psi_-(t) <B(0)<\lim\limits_{t \downarrow0} \psi_+(t)$, 
and hence, from~\eqref{eq:BPrice}, $B'$ is positive in the neighbourhood of the origin.
Moreover, for all $t>0$, Equation~\eqref{eq:Psipm} implies that 
$\psi_+'(t)<0<\psi_-'(t)$.
Let 
$t_0:=\sup\{t>0: \psi_-(s)<B(s)<\psi_+(s), \text{ for all }0<s<t\}$.
If $t_0$ is finite, then $B(t_0)\geq \psi_+(t_0)$ because $\psi_-(t)<0$ and $B$ is positive and increasing in $(0,t_0)$.
Therefore, $B(t_0)=\psi_+(t_0)$ and $B(t)\geq \psi_+(t)$ for all $t>t_0$ by comparison principle,
and hence $B$ is decreasing on $(t_0,+\infty)$. 
Since it is bounded below by $\lim\limits_{t \uparrow\infty} \psi_+(t)=0$,
it converges at infinity to some constant $C\geq 0$.
From the Riccati equation~\eqref{eq:ABfracH}, $B'$ then converges to $-\kappa C-\frac{1}{2}\xi^2 C^2$,
and necessarily $C=0$. 
If $t_0$ is infinite, then~$B$ is increasing and bounded from above by~$\psi_+$, 
which tends to $0$ at infinity;
this yields a contradiction since $B$ is increasing, so that $\lim\limits_{t \uparrow\infty} B(t)=0$.
Finally, since~$S$ is a martingale and the moment generating function is convex, $[0,1]\subseteq \mathcal{D}_t$ for all $t>0$ and hence $[0,1]\subseteq \mathcal{D}_{\infty}$.

\subsection{Proof of Theorem~\ref{thm:ImpliedVol}(i)}\label{sec:ProofSmallVol}
In this section and the next, the process $(X_t^{\BS})_{t\geq 0}$ shall denote the unique strong solution
starting from the origin to the Black-Scholes stochastic differential equation 
$\D X_t^{\BS} = -\frac{1}{2}\Sigma^2 \D t + \Sigma \D B_t$, 
for some given $\Sigma>0$, for $t>0$.
In this model, the price of a European Call option with maturity $t$ and strike $\E^{x}$ ($x\in\RR$) is given by
$$
C^{\BS}(x,t,\Sigma) = \Nn\left(-\frac{x}{\Sigma\sqrt{t}} + \frac{\Sigma\sqrt{t}}{2}\right)
-\E^{x}\Nn\left(-\frac{x}{\Sigma\sqrt{t}} - \frac{\Sigma\sqrt{t}}{2}\right),
$$
where $\Nn$ denotes the Gaussian cumulative distribution function.
Straightforward computations yield
$\log\EE\left( \E^{uX^{\BS}_t} \right) = \displaystyle \frac{1}{2}u(u-1)\Sigma^2t$, 
for all $u\in\RR$. 
In~\cite{FJ10}, the authors proved that the process $(X_t^{\BS})_{t\geq0}$ 
satisfies a large deviations principle with speed~$t^{-1}$ 
and good rate function $x\mapsto x^2/(2\Sigma^2)$,
which implies that the following limits hold:
$$
\lim_{t\downarrow 0}t\log \EE\left(\E^{x}-\E^{X_t^{\BS}}\right)_+
 = -\frac{x^2}{2\Sigma^2}
\quad \text{for } x\leq 0, 
\qquad\text{and}\qquad
\lim_{t\downarrow 0}t\log\EE\left(\E^{X_t^{\BS}}-\E^{x}\right)_+
 = -\frac{x^2}{2\Sigma^2}
\quad \text{for } x\geq 0.
$$
When $d\in (0,1/2)$, from Theorem~\ref{thm:LDP}(i), 
we can mimic this proof to obtain
$$
\lim_{t\downarrow 0}t\log \EE\left(\E^{x}-\E^{X_t}\right)_+
 = -\lambda^*_+(x),
\quad \text{for } x\leq 0, 
\qquad\text{and}\qquad
\lim_{t\downarrow 0}t\log\EE\left(\E^{X_t}-\E^{x}\right)_+
 = -\lambda^*_+(x),
\quad \text{for } x\geq 0,
$$
where $\lambda_\pm^*$ are defined in~\eqref{eq:RateFunctions}, so that, for any real number~$x$, $\Sigma(x,t)$ converges to~$\sqrt{\eta}$
as~$t$ tends to zero.
Likewise, in the case $d\in(-1/2,0)$, from Theorem~\ref{thm:LDP}(i)
we obtain
$$
\lim_{t\downarrow 0}t^{1+d}\log \EE\left(\E^{x}-\E^{X_t}\right)_+
 = -\lambda^*_-(x),
\quad \text{for } x\leq 0,
\qquad\text{and}\qquad
\lim_{t\downarrow 0}t^{1+d}\log\EE\left(\E^{X_t}-\E^{x}\right)_+
 = -\lambda^*_-(x),
\quad \text{for } x\geq 0.
$$
Consider the ansatz $\Sigma_t(x)=\sigma_0 t^{d/2}$, for some $\sigma_0>0$;
the Black-Scholes call price then reads
$$
C^{\BS}(x, t, \Sigma_t(x))
 = \Nn\left(\frac{\sigma_0 t^{1/2+d/2}}{2}-\frac{x}{\sigma_0 t^{1/2+d/2}}\right)
  - \E^x \Nn\left(-\frac{\sigma_0 t^{1/2+d/2}}{2}-\frac{x}{\sigma_0 t^{1/2+d/2}}\right).
$$

Since
$\Nn(z) = \E^{-z^2/2}\left(\frac{1}{z} - \frac{1}{2z^3} + o(z^{-3}) \right)$
as $z$ tends to minus infinity, we obtain, after simplifications,
$$
C^{\BS}(x, t, \Sigma_t(x))
 = \exp\left(-\frac{1}{2}\left[\frac{x^2}{\sigma_0 ^2 t^{1+d}}-x+\frac{1}{4}\sigma_0^2 t^{1+d}\right]\right)
 \left(\frac{\sigma_0 ^2 t^{3/2+3d/2}}{2x} + o\left(t^{3/2+3d/2}\right) \right).
$$
Therefore, taking $\sigma_0=\frac{x}{\sqrt{2\lambda_-^*(x)}}=\sqrt{\frac{v_0}{\Gamma(2+d)}}$, 
we obtain 
$\lim_{t\downarrow 0}t\log\EE (\E^{X^{\BS}_t}-\E^{x})_+ = -\lambda^*_+(x)$,
for all $x>0$.
Similarly,
$\lim_{t\downarrow 0}t\log\EE (\E^{x}-\E^{X^{\BS}_t})_+ = -\lambda^*_+(x)$, 
for all $x<0$.

\subsection{Proof of Theorem~\ref{thm:ImpliedVol}(ii)}\label{sec:ProofLargeVol}
Consider a Call option in the Black-Scholes model, 
with log strike~$x t^{1+d/2}$ 
and time-dependent implied volatility
$\widetilde{\sigma}(x,t) \equiv \sqrt{2}t^{d/4}\left( \sqrt{\Lambda_+^*(x)}+\sqrt{\Lambda_+^*(x)-x} \right)$. 
Then
 \begin{align*}
C^{\BS}\left(x, t, \widetilde{\sigma}(x,t)\right)
 &= \EE\left(\E^{X_t^{BS}}-\E^{xt^{1+d/2}}\right)_+
  = \Nn\left(-\frac{x t^{1+d/2}}{\widetilde{\sigma}(x,t)\sqrt{t}} + \frac{\widetilde{\sigma}(x,t)\sqrt{t}}{2}\right)
   - \E^{x t^{1+d/2}}\Nn\left(-\frac{x t^{1+d/2}}{\widetilde{\sigma}(x,t)\sqrt{t}} - \frac{\widetilde{\sigma}(x,t)\sqrt{t}}{2}\right) 
\\
  & = \Nn\left(t^{1/2+d/4}\sqrt{2(\Lambda_+^*(x)-x)}\right)
  - \E^{xt^{1+d/2}}\Nn\left(-t^{1/2+d/4}\sqrt{2\Lambda^*_+(x)}\right),
\end{align*}
where we used the identity
$
\sqrt{\Lambda_+^*(x)}+\sqrt{\Lambda_+^*(x)-x}
 + \frac{x}{\sqrt{\Lambda_+^*(x)} + \sqrt{\Lambda_+^*(x)-x}}
 = 2\sqrt{\Lambda^*_+(x)}$
in the second line.
Since $\Nn(y)=1-\E^{-y^2/2}\left(y^{-1}+o(y^{-1})\right)$ for large~$y$, we obtain
$$
1-C^{\BS}\left(x, \widetilde{\sigma}(x,t)\right)
 = \frac{\exp\left((x-\Lambda^*_+(x))t^{1+d/2}\right)}{\sqrt{2}t^{1/2+d/4}}\left(\frac{1}{\sqrt{\Lambda^*_+(x)}}+\frac{1}{\sqrt{\Lambda^*_+(x)-x}}+o(1)\right),
$$
and therefore 
$\lim_{t\uparrow\infty} t^{-(1+d/2)}
\log\left(1-C^{\BS}\left(x,t, \widetilde{\sigma}(x,t)\right)\right)=x-\Lambda^*_+(x)$.

Let us now return to the proof of the theorem, following the lines of~\cite[Theorem 13]{JKRM}.
Let~$f$ be a function diverging to infinity at infinity such that that the pointwise limit 
$\Lambda(u):=\lim_{t\uparrow\infty}f(t)^{-1}\log \EE\left(\E^{uX_t}\right)$
exists for all $u\in \Dd_\Lambda\subset\RR$.
Since the stock price is a true positive martingale, we can define a new probability measure $\widetilde{\PP}$
via $\left.\D\widetilde{\PP}/\D\PP\right|_{\Ff_t} = S_t$.
Under~$\widetilde{\PP}$, the limiting cumulant generating function of~$(X_t/t)_{t\geq 0}$ reads
$\widetilde{\Lambda}(u) :=  \lim_{t\uparrow\infty}f(t)^{-1}\log \widetilde{\EE}\left(\E^{uX_t}\right)$,
and clearly
$\widetilde{\Lambda}(u) = \Lambda(u+1)$ for all 
$u\in \Dd_{\widetilde{\Lambda}}=\{v\in\RR: 1+v\in\Dd_\Lambda\}$.
Note that $0\in\Dd_{\widetilde{\Lambda}}^\circ$ if and only if  $1\in\Dd_{\Lambda}^\circ$.
This identity also shows that the Fenchel-Legendre transforms are related via
$\widetilde{\Lambda}^*(x)=\Lambda^*(x)-x$, for all $x\in\RR$.
Now, if the family $(X_t/f(t))_{t\geq1}$ satisfies a large deviations principle 
under both~$\PP$ and~$\PPt$ with speed $f(t)$ 
and good rate functions~$\Lambda^*$ and~$\widetilde{\Lambda}^*$, then the following behaviours hold:
\begin{eqnarray}\label{eq:OptionPricesLarge}
\begin{array}{llrl}
\quad & \text{(Put option)}  & 
\displaystyle \lim_{t\uparrow\infty}f(t)^{-1}\log\EE\left(\E^{xf(t)}-\E^{X_t}\right)_+  & = 
\left\{
\begin{array}{ll}
x-\Lambda^*(x) & \text{if }x\leq x^*,\\
x & \text{if }x>x^*,
\end{array}
\right.
\\
\quad & \text{(Call option)}  & 
\displaystyle \lim_{t\uparrow\infty}f(t)^{-1}\log\EE\left(\E^{X_t}-\E^{xf(t)}\right)_+  & = 
\left\{
\begin{array}{ll}
-\wt\Lambda^*(x)\quad & \text{if }x\geq \wt x^*,\\
0\quad & \text{if }x<\wt x^*,
\end{array}
\right.
\\
\quad & \text{(covered Call option)}  & \displaystyle \lim_{t\uparrow\infty}f(t)^{-1}\log\left[1-\EE\left(\E^{X_t}-\E^{xf(t)}\right)_+\right] & =
\left\{
\begin{array}{ll}
0 & \text{if }x> \wt x^*,\\
x-\Lambda^*(x)  &\text{if }x\in\left[x^*,\wt x^*\right],\\
x\quad & \text{if }x< x^*,
\end{array}
\right.
\end{array}
\end{eqnarray}
where $x^*$ and $\wt x^*$ are where the functions $\Lambda^*$ and $\widetilde{\Lambda}^*$
attain their minimum, and satisfy $x^*=\Lambda'_+(0) \leq \Lambda'_-(1) = \wt x^*$.
Furthermore the convergence in (i)-(iii) is uniform in $x$ on compact subsets of $\RR$.

We first prove the theorem in the case $d\in (0,1/2]$. 
For $f(t)\equiv t^{1+d/2}$, Equations~\eqref{eq:OptionPricesLarge} imply
\begin{equation*}
\begin{array}{rll}
\lim_{t\uparrow\infty}t^{-(1+d/2)}\log\EE\left(\E^{xt^{1+d/2}}-\E^{X_t}\right)_+
  & = x,
 &\text{for all } x\in \RR, \\
\lim_{t\uparrow\infty}t^{-(1+d/2)}\log\EE\left(\E^{X_t}-\E^{xt^{1+d/2}}\right)_+
 & =0,
 & \text{for all }x\in \RR, \\
\lim_{t\uparrow\infty}t^{-(1+d/2)}\log\left(1-\EE\left(\E^{X_t}-\E^{xt^{1+d/2}}\right)_+\right)
 & = x-\Lambda^*(x),
 & \text{for all }x\in \RR.
\end{array}
\end{equation*}
Therefore, the implied volatility satisfies 
$\lim_{t\uparrow\infty}t^{-d/4}\Sigma(x t^{1+d/2})
 = \sqrt{2}\left( \sqrt{\Lambda_+^*(x)}+\sqrt{\Lambda_+^*(x)-x} \right)$ 
 for all $x$ in $\RR$.
The proof of the theorem in the case $d\in(-1/2,0)$ is analogous to that in~\cite{JKRM},
and is therefore omitted.

\appendix
\section{The G\"artner Ellis Theorem}\label{app:LDP}
We provide here a brief review of large deviations and the G\"artner-Ellis theorem.
For a detailed account of these, the interested reader should consult~\cite{DZ}.
Let $(X_n)_{n\in\NN}$ be a sequence of random variables in~$\RR$, with law~$\mu_n$
and cumulant generating function $\Lambda_n(u) \equiv \log\EE(\E^{u X_n})$.
For a Borel subset $A$ of the real line, we shall denote respectively by 
$A^o$ and $\bar{A}$ its interior and closure (in $\RR$).

\begin{definition}\label{def:LDP}
The sequence {$X_n$} is said to satisfy a large deviations principle with speed $n$ 
and rate function~$I$ if for each Borel mesurable set~$E$ in~$\RR$, 
$$
-\inf_{x\in E^o} I(x)
 \leq \liminf_{n\uparrow \infty} \frac{1}{n}\log\PP\left(X_n \in E \right)
\leq \limsup_{n\uparrow \infty} \frac{1}{n}\log\PP\left(X_n \in E \right)
\leq -\inf_{x\in \bar{E}}I(x).
$$
Furthermore, the rate function is said to be good if it is strictly convex on the whole real line.
\end{definition}
Before stating the main theorem, we need one more concept:
\begin{definition}
Let $\Lambda : \RR \rightarrow (-\infty, +\infty]$ be a convex function, 
and $\Dd_\Lambda := \{u\in\RR: \Lambda(u) < \infty\}$ its effective domain. 
The function~$\Lambda$ is said to be essentially smooth if
\begin{itemize}
\item The interior $\Dd^o_\Lambda$ is non-empty;
\item $\Lambda$ is differentiable throughout $\Dd^o_\Lambda$;
\item $\Lambda$ is steep: $\lim \limits_{n\uparrow \infty}|\Lambda' (u_n)|=\infty$ 
for any sequence $(u_n)_{n\geq 1}$ in $\Dd_\Lambda^o$ converging to a boundary point 
of~$\Dd_\Lambda^o$.
\end{itemize}
\end{definition}
Assume now that the limiting cumulant generating function
$\Lambda(u) := \lim_{n\uparrow \infty }n^{-1}\Lambda_n(nu)$
exists as an extended real number for all $u\in\RR$, and let $\Dd_\Lambda$ denote its effective domain.
Let $\Lambda^*:\RR\to\RR_+$ denote its (dual) Fenchel-Legendre transform, via the variational formula
$\Lambda^*(x) := \sup_{\lambda \in D_\Lambda}\{\lambda x-\Lambda(\lambda)\}$.
Then the following holds:
\begin{theorem}[G\"artner-Ellis theorem]\label{thm:GE}
If $0\in \Dd_{\Lambda}^o$, $\Lambda$ is lower semicontinuous and essentially smooth,
then the sequence $(X_n)_n$ satisfies a large deviations principle with rate function~$\Lambda^*$.
\end{theorem}
When partial conditions of the G\"artner-Ellis theorem are satisfied, a full large deviations principle
might not be available, but one can define a partial one, as follows:
\begin{definition}\label{def:partialLDP}
We shall say that the sequence~$(X_n)$ satisfies a partial large deviations principle with speed~$n^{-1}$
on some interval~$I'$ with rate function~$\Lambda^*$ if Definition~\ref{def:LDP} holds
for all subsets~$A\subset I'$.
\end{definition}
It is clear that if $0\in \Dd_{\Lambda}^o$ and $\Lambda$ is lower semicontinuous and strictly convex 
on some interval~$I\subset\RR$, 
then~$(X_n)$ satisfies a partial large deviations principle 
on~$I' = \Lambda'(I)$ with rate function $\Lambda^*(x):=\sup_{u\in I}\{ux-\Lambda(u)\}$.


\end{document}